\documentclass[12pt,reqno,a4paper,twoside]{article}

\usepackage{amsmath,amsthm,amstext,amscd,amssymb,euscript,mathrsfs}
\usepackage{graphics,color,calc,graphicx}
 \usepackage[utf8x]{inputenc}
\usepackage{calrsfs}

\usepackage{epsfig}

\renewcommand{\phi}{\varphi}

\def\1{{\mathchoice {\rm 1\mskip-4mu l} {\rm 1\mskip-4mu l}
{\rm 1\mskip-4.5mu l} {\rm 1\mskip-5mu l}}}

\newtheorem{theorem}{{\small T}{\scriptsize HEOREM}}[section]
\newtheorem{corollary}{{\bf{\small C}{\scriptsize OROLLARY}}}[section]
\newtheorem{proposition}{{\bf{\small P}{\scriptsize ROPOSITION}}}[section]
\newtheorem{lemma}{{\bf{\small L}{\scriptsize EMMA}}}[section]
\newtheorem{remark}{{\bf{\small R}{\scriptsize EMARK}}}[section]
\newtheorem{definition}{{\bf{\small D}{\scriptsize EFINITION}}}[section]

\renewenvironment{proof}[1]
{\noindent{{\bf{\small{ P}{\scriptsize ROOF}}}.}\hspace{0.1cm} #1} {$\;\qed$\newline}

\newcommand{\beq}{\begin{eqnarray}}
\newcommand{\eeq}{\end{eqnarray}}

\newcommand{\ba}{\begin{align*}}
\newcommand{\ea}{\end{align*}}

\newcommand{\be}{\begin{equation}}
\newcommand{\ee}{\end{equation}}

\newcommand{\bl}{\begin{lemma}}
\newcommand{\el}{\end{lemma}}

\newcommand{\br}{\begin{remark}}
\newcommand{\er}{\end{remark}}

\newcommand{\bt}{\begin{theorem}}
\newcommand{\et}{\end{theorem}}

\newcommand{\bd}{\begin{definition}}
\newcommand{\ed}{\end{definition}}

\newcommand{\bp}{\begin{proposition}}
\newcommand{\ep}{\end{proposition}}

\newcommand{\bc}{\begin{corollary}}
\newcommand{\ec}{\end{corollary}}

\newcommand{\bpr}{\begin{proof}}
\newcommand{\epr}{\end{proof}}

\newcommand{\bi}{\begin{itemize}}
\newcommand{\ei}{\end{itemize}}

\newcommand{\ben}{\begin{enumerate}}
\newcommand{\een}{\end{enumerate}}

\title{Multilinearity of two-point correlation functions in one-dimensional models out of equilibrium}
\author{
Frank Redig$^{\textup{{\tiny(a)}}}$,\\
Wioletta Ruszel$^{\textup{{\tiny(a)}}}$,
\\
{\small $^{\textup{(a)}}$
Delft University of Technology}, {\small Mekelweg 4 2628 CD Delft , The Netherlands}
\\
}
\begin{document}
\maketitle
%\tableofcontents

\section{Introduction}
One of the main issues of models in equilibrium statistical mechanics is to explain phase transitions and critical phenomena. In some sense the macroscopic behaviour does not depend much on the microscopic one which leads to universal
behaviour for different systems at large scales.

In non-equilibrium systems a greater variety of phenomena occurs and in some sense is more sensitive to microscopic information. A unified thermodynamical theory is much more difficult to obtain and different attempts were made.

The simplest situation where one can create a so-called non-equilibrium steady states (NESS) is to consider a system
coupled at the boundaries in with different reservoirs creating currents. To construct such models, one usually starts from
a simple bulk model to which boundary terms are added. By simple bulk model here we mean a reversible Markovian conservative
dynamics (such as the simple symmetric exclusion process) having simple
stationary measures e.g. of product nature. Upon coupling such a system to different reservoirs, the nature of the stationary measures
changes dramatically, i.e., NESS thus created are far from product, and generically show long-range correlations.
In some special systems, these correlations are accessible analytically \cite{spohn}. Such models could be considered as the non-equilibrium analogue of
exactly solvable models in equilibrium (such as the Ising model). On the macro scale in a large class of one-dimensional
conservative Markov dynamics, it is believed \cite{Gab} that two-point correlation functions in NESS
are multilinear. On the micro scale this multi-linearity is rare and a possible indication of exact solvability e.g. in the sense of matrix ansatz solution
\cite{Der}. 

In \cite{Dual}  a class of wealth distribution models with two agents in equilibrium is considered. The authors showed under which conditions on the redistribution measure and transition operator there exist stationary product measures and duality functions. At each time step two agents are redistributing a random amount $\epsilon$ of their wealth and a constant part $\lambda \in [0,1]$ is retained.
It turns out that there exist only stationary product measures for the trivial case $\lambda=0$ and the law of $\epsilon$ is independent of the amount of total wealth and in particular is $Beta(k,k)$ distributed.

In \cite{Gab} the authors study particle systems out of equilibrium. Multilinearity of the two-point functions is obtained for the KMP (Kipnis-Marchioro-Presutti) model of heat conduction and boundary driven SEP (symmetric exclusion process) for  trivial $\lambda$, $\epsilon$ uniform and the reservoirs follow an exponential distribution. The two-point correlation functions turn out to be negatively correlated for the SEP resp. positively correlated for KMP while the off-diagonal terms are essentially the same in the large $N$ limit.

An example where one does not obtain multilinearity for a model of heat conduction can be found in \cite{kia}.
The authors study the Brownian momentum process which is weakly coupled
to heat baths and in particular the NESS  and its proximity to the local equilibrium measure
in terms of the strength of coupling. For three- and four-site systems, they obtain
the two-point correlation function and show it is generically not multilinear.

Here we are interested in the question under which conditions we have multilinear two-point functions in this class of models.
The conditions will only include a relation between the moments of the reservoir laws, moments of the redistribution parameter and the constant $\lambda$.
The result presents a first example where one can obtain multilinear two-point functions even in the absence of product stationary measures.

The rest of our paper is organized as follows.
In section \ref{sec:def} we will provide all necessary definitions. Section \ref{sec:res} deals with the main result and finally we discuss in section 4 some generalizations. Detailed computations can furthermore be found in the appendix.

\section{Notation and Definitions}\label{sec:def}

We consider one-dimensional interacting models on the set $\{1,...,N \}$ coupled to some reservoirs. The sites $1,...,N$ can be interpreted as \textit{particles} or \textit{agents}. The boundary sites are interacting with reservoirs which are represented by the \textit{ghost sites} $0$ and $N+1$. Let $\Omega=S^N=[0,\infty)^N$ denote the state space of the model. Each element $(x_1,...,x_N)\in \Omega$ can be seen as e.g. the \textit{wealth} or \textit{energy} at vertex $1,...,N$, see \cite{Dual,KMP}.

For some random $\epsilon \in [0,1]$ and constant $\lambda \in (0,1)$ we can define a map $T^{k,l}_{\lambda,\epsilon} : \Omega \rightarrow \Omega$ for some $k,l \in \{ 1,..., N\}$ with $k<l$ by
\begin{equation}\label{defT}
\begin{split}
& T^{k,l}_{\lambda,\epsilon}(x_1,...,x_N) = \\
&(x_1,...,x_{k-1},\lambda x_k + \epsilon(1-\lambda)(x_k+x_l),...,\lambda x_l +  (1- \epsilon)(1-\lambda)(x_k+x_l),...,x_N)
\end{split}
\end{equation}
At each time step some random amount of \textit{energy} or \textit{wealth} is exchanged.
During the exchange some fixed amount $\lambda$  is kept while some random amount (modeled by $\epsilon$) is exchanged.

We remark that the map conserves the total wealth (energy) $s=x_1+...+x_N$.

Let $\nu$ be the law of  $\epsilon \in [0,1]$ which will be called the redistribution parameter. We further assume (A1) that $\nu$ is symmetric and its first moment is equal to $\frac{1}{2}$ and that it has a second moment (A2) which will be abbreviated by $\int_0^1 \epsilon(1-\epsilon) \nu(d\epsilon) =:\alpha$.  Note that it follows trivially that by (A1), $\alpha \leq \frac{1}{4}$.

 After some exponential waiting time with mean 1, the initial state $(x_1,...,x_N)\in \Omega$ is is redistributed and replaced by $ T^{k,l}_{\lambda,\epsilon}(x_1,...,x_N) $.

The generator of the model can be expressed as the sum of the bulk generator describing the dynamics in the bulk and boundary generators which describe the interaction with the reservoirs,
\begin{equation}\label{defL}
\mathcal{L} = \mathcal{L}_L + \mathcal{L}_b + \mathcal{L}_R.
\end{equation}
It is a Markov process.
Let $f$ be some  bounded continuous function, $f:\Omega \rightarrow \mathbb{R}$
then
\begin{equation}
\begin{split}
\mathcal{L}_b(&f(x_1,...,x_N)) = \\
&\sum_{k,l\in \{2,...,N-1\}} p(k,l) \biggl[\int_0^1 f(T^{k,l}_{\lambda,\epsilon}(x_1,...,x_N)) \nu(d\epsilon) -f(x_1,...,x_N)\biggr ]
\end{split}
\end{equation}
where $p(k,l)$ is the transition probability of a symmetric nearest neighbour random walk given by
\[
p(k,l) = \frac{1}{2}\delta_{k,k-1} + \frac{1}{2}\delta_{k,k+1}
\]
for $k,l\in \{2,...,N-1\}$.
The generators of the reservoirs are defined as
\begin{equation}\label{defLeft}
\begin{split}
\mathcal{L}_L( &f(x_1,...,x_N)) =  \\
&  \int_0^{\infty}\int_0^1 f(\lambda x_1 + \epsilon(1-\lambda)(x_0 + x_1),x_2,..,x_N) \nu(d\epsilon) \mu_L(dx_0) -f(x_1....,x_N)
\end{split}
\end{equation}
resp.
\begin{equation}
\begin{split}\label{LRight}
&\mathcal{L}_R(f(x_1,...,x_N)) =  \\
 \int_0^{\infty} \int_0^1 & f(x_1,...,x_{N-1},\lambda x_N + \epsilon(1-\lambda)(x_N + x_{N+1})) \nu(d\epsilon) \mu_R(dx_{N+1}) \\
& -f(x_1,...,x_N)
\end{split}
\end{equation}
$\mu_0$ resp. $\mu_{N+1}$ denote the distributions of the reservoirs at the ghost sites $0$ resp. $N+1$.
We abbreviate their first moments  by
\begin{equation}
T_L := \int_0^{\infty} x_0 \mu_L(dx_0) \text{ ,\space}  T_R := \int_0^{\infty} x_{N+1} \mu_{R}(dx_{N+1})
\end{equation}
resp. the second moments by
\begin{equation}\label{secmom}
L^2 := \int_0^{\infty} x^2_0 \mu_L(dx_0) \text{ ,\space}  R^2 := \int_0^{\infty} x^2_{N+1} \mu_{R}(dx_{N+1})
\end{equation}
We remark that at this point we only assume that the first and second moments of
the reservoir measures exist.

We call a probability measure $\mu \in \mathcal{P}(\Omega)$ on the set of probability measures on $\Omega$ \textit{stationary} if and only if for $f$ bounded and continuous on $\Omega$:
\begin{equation}\label{defSt}
\int_{\Omega} \mathcal{L}(f(x_1,...,x_N))\mu(dx_1,...,dx_N) = 0.
 \end{equation}
For $i<j$ and $i,j \in \{ 0,...,N+1 \}$ the two-point function $\mu(x_i,x_j)$ is equal to $\mu(x_i,x_j):=\int_{\Omega}x_i  x_j \mu(dx_1,...,dx_N)$. Later we will use for convenience the notation $\mu_{ij}$. Further let
\begin{equation}
C_N(i,j) :=\mu(x_i,x_j)-E_N(i)E_N(j)
\end{equation}
denote the \textit{two-point correlation function} of $\mu$. We set
\begin{equation}
C_N(0,0):=\int_0^{\infty} (x_0-T_L)^2\mu_0(dx_0)
\end{equation}
resp. for
\begin{equation}
C_N(N+1,N+1):=\int_0^{\infty} (x_{N+1}-T_R)^2\mu_{N+1}(dx_{N+1}).
\end{equation}
Note that this is equivalent to saying $L^2=:\mu(x_0,x_0)$ resp.  $R^2=:\mu(x_{N+1},x_{N+1})$. For $0\leq i \leq N$ and $1\leq j\leq N+1$ we set $C_N(0,j):=C_N(i,N+1):=0$.

We say that the two-point function satisfies the \textit{multilinearity ansatz} if and only if we can find some coefficients $a,b,c,d,e,f,g$ such that
 for $i,j \in \{ 0,...,N+1\}$ and $i<j$
\begin{equation}\label{multi}
\begin{cases}
 \mu_{i,j} = a + b i + c j + d ij  & \text{ if } i< j, i=1,...,N-1, j= 2, ..., N\\
 \mu_{i,i}= e + f i + g i^2 & \text{ if } i=j, i=1,...,N.
\end{cases}
\end{equation}

%%%%%%%%%%%%%%%%%%%%%%%%%%%%%%%%%%
\section{Result}\label{sec:res}
We present our main theorem.
\begin{theorem}
The two-point functions $(\mu_{i,j})_{i,j}$ are multilinear with coefficients given by
\begin{equation}
\begin{split}
& a:=T^2_L \\
& b := \frac{\left(T_R-T_L\right) \left[(N+1) (\lambda +2\alpha (1- \lambda)) T_L+(1-\lambda) (1-4\alpha)T_R\right]}{(N+1)(1+N \lambda+2(N-1) (1-\lambda)\alpha)}\\
& c:= \frac{(T_R-T_L)T_L}{N+1} \\
& d:=\frac{(\lambda + 2 \alpha (1-\lambda)) \left(T_L-T_R\right){}^2}{(N+1) (1+N \lambda + 2(N-1)  (1-\lambda)\alpha)}\\
& f:= \frac{(1-2 \alpha (1-\lambda)) \left(T_R-T_L\right) \left[(1+ (2N+1) \lambda +4 N \alpha (1-\lambda) ) T_L+(1 - 4 \alpha) (1-\lambda) T_R\right]}{(N+1) (1 + N \lambda + 2(N-1)  (1-\lambda) \alpha ) (\lambda + 2 \alpha (1-\lambda))} \\
& g:= \frac{(1-2 \alpha (1-\lambda)) \left(T_L-T_R\right){}^2}{(N+1) (1+N \lambda + 2(N-1) (1-\lambda)\alpha)}.
\end{split}
\end{equation}
if and only if  $L^2 = L^2(\alpha,\lambda) $ and $R^2=R^2(\alpha,\lambda)$ are chosen in the following way
\begin{equation}\label{LR}
\begin{split}
 L^2  &:=\frac{(1-2\alpha(1-\lambda)) }{(\lambda+2\alpha(1-\lambda))} T^2_L\\
& + \frac{\alpha(1-\lambda)(1-2\alpha(1-\lambda))(T_L-T_R)^2}{(N+1)[1+N\lambda+2(N-1)(1-\lambda) \alpha](\lambda+2\alpha(1-\lambda))} \\
 R^2 &:=\frac{(1-2\alpha(1-\lambda)) }{(\lambda+2\alpha(1-\lambda))} T^2_R\\
&5 + \frac{\alpha(1-\lambda)(1-2\alpha(1-\lambda))(T_L-T_R)^2}{(N+1)[1+N\lambda+2(N-1)(1-\lambda) \alpha](\lambda+2\alpha(1-\lambda))} \end{split}
\end{equation}

In particular the two-point correlation functions are equal to
\begin{equation}
C_N(i,j)=
\begin{cases}
&\biggl ( \frac{(1-4 \alpha) (1-\lambda) \left(T_L-T_R\right){}^2}{ (1+\lambda N +2  (N-1) (1-\lambda) \alpha )}\biggr ) \frac{i}{N+1}\biggl (1 - \frac{j}{N+1} \biggr ) : 0\leq i<j\leq N+1 \\
& \frac{(1- 4 \alpha) (1-\lambda) T_L^2}{\lambda + 2 \alpha (1-\lambda)} + \frac{\alpha (1-2 \alpha (1-\lambda)) (1-\lambda) \left(T_L-T_R\right){}^2}{(\lambda + 2 \alpha (1-\lambda))  [1+\lambda N + 2 (N-1)  (1-\lambda) \alpha ](N+1)}  \\
&+ \biggl ( \frac{(1-4 \alpha) (1-\lambda) \left(T_R-T_L\right) \left((1+2 \lambda N+2 \alpha (2 N-1) (1-\lambda) ) T_L + (1-2 \alpha (1-\lambda)) T_R\right)}{( \lambda + 2 \alpha (1-\lambda)  ) (1+ \lambda N+ 2  (N-1) (1-\lambda) \alpha)}\biggr ) \frac{i}{N+1} \\
&+ \biggl ( \frac{(1-4 \alpha) (1-\lambda)  \left(T_L-T_R\right){}^2 N}{ (1+ \lambda N + 2  (N-1) (1-\lambda)  \alpha )} \biggr ) \frac{i^2}{(N+1)^2}: 0\leq i=j \leq N+1\\
\end{cases}
\end{equation}
\end{theorem}
\begin{proof}
%%%%%%%%%%%%%%%%%%%%%%%%%%%%%%%%%%

We are looking for expressions for $(\mu_{i,j})_{i,j}$ for $i,j \in \{ 0,...,N+1\}$.
In the bulk for $i\leq j$,  $\mu_{i,j}$ have to satisfy the following set of equations (see appendix):
\begin{equation}
\begin{cases}\label{eqbulk}
-2(1-A) \mu_{ii} + A( \mu_{i-1,i-1} +  \mu_{i+1,i+1}) + B (\mu_{i,i+1} + \mu_{i-1,i}) =  0 & \text{ if } i=2,...,N-1\\
\mu_{i+1,j} + \mu_{i-1,j} + \mu_{i,j+1} + \mu_{i,j-1}- 4 \mu_{i,j}  =  0 & \text{  if } {i=2,..,N-3 \atop j=4,...,N-1}\\
C(\mu_{i,i} +\mu_{i+1,i+1}) +\frac{1}{2}(\mu_{i-1,i+1} + \mu_{i,i+2}) -(1+B)\mu_{i,i+1}  =  0 &\text{ if } i=2,..,N-2
\end{cases}
\end{equation}
with $A,B, C$ defined in the previous section as
\begin{equation}
\begin{split}
A & = \biggl(\frac{1}{2}-\alpha \biggr)(1-\lambda)\\
B & = 1-2\alpha (1-\lambda) \\
C & = \frac{\lambda}{2} + \alpha(1-\lambda)
\end{split}
\end{equation}
On the boundaries we need the following conditions to be satisfied.
First for the left reservoir ($i=1$):
\begin{equation}
\begin{cases}\label{eqL}
-2(1-A) \mu_{11} + A(L^2  +  \mu_{22}) + B (\mu_{1,2} + T_L E_N(1)) =  0 & \text{ if } j=1\\
\mu_{2,j}+T_L E_N(j) + \mu_{1,j-1} + \mu_{1,j+1}-4\mu_{1j} =  0 & \text{  if }j=3,..,N-1\\
C(\mu_{11}+\mu_{22} )+ \frac{1}{2}(\mu_{13}+T_L E_N(2)) -(1+B)\mu_{12}=0 & \text{  if }j=2
\end{cases}
\end{equation}
and second on the right for $j=N$,
\begin{equation}
\begin{cases}\label{eqR}
-2(1-A) \mu_{NN} + A(\mu_{N-1,N-1} + R^2 ) + B (\mu_{N-1,N} + T_R E_N(N)) =  0 & \text{ if } i=N\\
\mu_{i+1,N}+ \mu_{i-1,N} + T_R E_N(i) + \mu_{i,N-1}- 4\mu_{iN} =  0 & \text{  if } i=2,.., N-2\\
C(\mu_{N-1,N-1}+\mu_{NN})+\frac{1}{2}(\mu_{N-2,N}+T_RE_N(N-1))-(1+B)\mu_{N-1,N} =  0 &\text{ if } i=N-1
\end{cases}
\end{equation}
and
\begin{equation}
E_ N(i)=T_L \biggl ( 1- \frac{i}{N+1}\biggr ) + T_{R}\frac{i}{N+1} \text{ , }
 i=0,...,N+1.
\end{equation}
The result follows from some tedious computations.
\end{proof}
%%%%%%%%%%%%%%%%%%%%%%%%%%%%%%%%%%

\section{Conclusion and Discussion}
In this note we obtained conditions under which two-point functions in the non-equilibrium case will be multilinear for some class of one-dimensional interacting models of wealth (heat) transport.
In particular we present a first example of models for which the two-point function is multilinear when the stationary measure is not product.

Let us make some final remarks. \\
\textbf{Remark I:} For $\lambda=0$ and $\alpha=\frac{1}{6}$ we find the two-point correlation functions obtained for the KMP model in \cite{Gab}. \\
\textbf{Remark II:} For a non-degenerate distribution of the redistribution parameter $\epsilon$, $\alpha < \frac{1}{4}$, Hence the occupation variables are always positively correlated. If $\nu=\delta_{1/2}$ is a degenerate measure, $\alpha = \frac{1}{4}$, and the correlation functions are 0.\\
\textbf{Remark III:} We tried to apply the multilinearity ansatz in the case that at the reservoirs are additionally depending on factors $\gamma_L$ and $\gamma_R$, hence the temperature profile is not entirely linear. The boundary generators are then given by
\begin{equation}
\begin{split}
&\mathcal{L}_L( f(x_1,...,x_N)) =  \\
& \gamma_L \int_0^{\infty}\int_0^1 f(\lambda x_1 + \epsilon(1-\lambda)(x_0 + x_1),x_2,..,x_N) \nu(d\epsilon) \mu_L(dx_0) -f(x_1....,x_N)\\
&\mathcal{L}_R(f(x_1,...,x_N)) =  \\
& \gamma_R \int_0^{\infty} \int_0^1 f(x_1,...,x_{N-1},\lambda x_N + \epsilon(1-\lambda)(x_N + x_{N+1})) \nu(d\epsilon) \mu_R(dx_{N+1})\\
&-f(x_1,...,x_N)
\end{split}
\end{equation}
It turns out that the two-point functions will never be multilinear as long as $\gamma_L , \gamma_R \neq 1$ or $T_L\neq T_R$. \\
\textbf{Remark IV:} Further we tried under what conditions three-point function might be multilinear. Unfortunately this problem is very complex and we could not find a general solution to this problem even for $N=6$.

\section{Appendix}

\subsection{Generator for the two-point function }
In the following we determine the generator $\mathcal{L}$ (see \eqref{defL}) of the two-point function $f_{ij}(x_1,...,x_N):=x_ix_j$. It will also depend on the 1-point linear functions on the borders $f_{i}(x_1,...,x_N):=x_i$. Recall that the transition operator $T^{k,l}_{\lambda,\epsilon}$ was defined in \eqref{defT} , $\epsilon$ satisfies (A1) and (A2) and the reservoirs have finite first and second moments..
We will distinguish 3 cases. Case I represents $i=j$,  case II $|i-j|>1$ and finally case II $|i-j|=1$.

\subsubsection{Temperature profile}

We will determine the closed form of the linear functions corresponding to the temperature profile of the system. Let $\mu$ be a stationary measure. We calculate the density profile $E_N(i):=\int_{\Omega}x_i \mu(dx_1,...,dx_N)$ by solving
\begin{equation}
\int_{\Omega} \mathcal{L}(x_i)\mu(dx_1,...,dx_N) = 0
\end{equation}
for $i=1,...,N$. Let $i=2,...,N-1$ then we need to solve
\[
\int_{\Omega} \mathcal{L}(x_i)\mu(dx_1,...,dx_N) = \frac{1-\lambda}{2}(E_N(i+1) +E_N(i-1)-2E_N(i)) =0
\]
for $i=1$ and $i=N$ we have
\begin{equation}
\begin{split}
& \int_{\Omega} \mathcal{L}(x_1)\mu(dx_1,...,dx_N) = \frac{1-\lambda}{2}(E_N(2) +T_L-2E_N(1)) =0 \\
& \int_{\Omega} \mathcal{L}(x_N)\mu(dx_1,...,dx_N) = \frac{1-\lambda}{2}(T_R +E_N(N-1)-2E_N(i)) =0.
\end{split}
\end{equation}
We can compute explicitely the closed form expression namely
\begin{equation}\label{mui}
\begin{split}
E_N(i) =  T_L \biggl ( 1- \frac{i}{N+1}\biggr ) + T_{R}\frac{i}{N+1} &\text{ , }
 i=0,...,N+1.
\end{split}
\end{equation}
Note that the linear profile is the same as in \cite{Gab} for the SEP and KMP models.

\subsubsection{Case I: $i=j$}
We will determine the generator
\[
\mathcal{L}(f_{ii}(x)) = \mathcal{L}_L(f_{11}(x)) + \mathcal{L}_b(f_{ii}(x)) + \mathcal{L}_R(f_{NN}(x))
\]
in the case that $i=j$.
Let us first fix $i=2,...,N-1$. First we calculate the generator in the bulk $\mathcal{L}_b(f_{ii}(x))$.

\begin{equation}
\begin{split}
\mathcal{L}_b(f_{ii}(x)) & = \sum_{k,l\in \{1,...,N\}} p(k,l) \biggl[\int_0^1 f(T^{k,l}_{\lambda,\epsilon}(x_1,...,x_N)) \nu(d\epsilon) -f(x_1,...,x_N)\biggr ]\\
& = \frac{1}{2}(f_{ii}(T^{i,i+1}(x))-f_{ii}(x)) + \frac{1}{2}(f_{ii}(T^{i,i-1}(x))-f_{ii}(x)) \\
& + \frac{1}{2}(f_{ii}(T^{i+1,i}(x))-f_{ii}(x)) + \frac{1}{2}(f_{ii}(T^{i-1,i}(x))-f_{ii}(x)) \\
& = f_{ii}(T^{i,i+1}(x)) + f_{ii}(T^{i,i-1}(x)) - 2f_{ii}(x)
\end{split}
\end{equation}
The boundary generators are obtained in the following way, let $i=1$:
\[
\mathcal{L}_L(f_{11}(x)) = f_{11}(T^{0,1}(x))-f_{11}(x) + f_{12}(T^{1,2}(x)) - f_{11}(x)
\]
where
\begin{equation}
\begin{split}
f_{11}(T^{0,1}(x)) & = \int_0^{\infty}\int_0^1 (\lambda x_1 + \epsilon (1-\lambda)(x_1+x_0)) \nu(d\epsilon) \mu_L(dx_0)\\
& = \lambda x_1 + \frac{1}{2}(1-\lambda)(x_1+T_L) \\
& = \frac{1+\lambda}{2}x_1 + \frac{1}{2}(1-\lambda)T_L
\end{split}
\end{equation}
and analogously for the right boundary generator $\mathcal{L}_R$ with $T_L$ replaced by $T_R$ and $x_1$ by $x_N$.
Hence for $i=1,...,N$  the generator $\mathcal{L}$ acting on $f_{ii}(x)$ is equal to
\[
\mathcal{L}(f_{ii}(x)) = \mathcal{L}_L(f_{11}(x)) + \mathcal{L}_b(f_{ii}(x)) +\mathcal{L}_R(f_{NN}(x))
\]
with
\begin{equation}\label{GeniEqualj}
\begin{split}
&\mathcal{L}_L(f_{11}(x))  = (1-\lambda)\biggl[-2(1-A) f_{1,1}(x) + A( L^2 +  f_{2,2}(x)) + B (f_{1,2}(x) +  T_L f_1(x)) \biggr ] \\
& \mathcal{L}_b(f_{ii}(x))  = \\
&(1-\lambda)\biggl[-2(1-A) f_{ii}(x) + A (f_{i-1,i-1}(x)  + f_{i+1,i+1}(x)) + B(f_{i,i+1}(x) + f_{i-1,i}(x) )\biggr]\\
& \mathcal{L}_R(f_{NN}(x))  = \\
&(1-\lambda)\biggl [-2(1-A) f_{N,N}(x) + A( f_{N-1,N-1}(x)+ R^2 ) + B (f_{N-1,N}(x) +  T_R f_N(x))\biggr]  \\
\end{split}
\end{equation}
with coefficients given by
\begin{equation}
\begin{split}
A & =A(\alpha,\lambda) = \biggl(\frac{1}{2}-\alpha \biggr)(1-\lambda)\\
B & =  B(\alpha,\lambda) = \lambda + (1-2\alpha)(1-\lambda)
\end{split}
\end{equation}

\subsubsection{Case II:  $|i-j| > 1$}
We calculate $\mathcal{L}(f_{ij}(x))$ for $f_{ij}(x)=x_ix_j$. Let us first assume $i < j$, $i=2,...,N-2$ and $j=4,...,N-1$. We can easily see doing similar calculations as in the first case that

\begin{equation}
\mathcal{L}_b(f_{ij}(x)) =  f_{ij}(T^{i,i+1}(x)) + f_{ij}(T^{i-1,i}(x))+ f_{ij}(T^{j,j+1}(x)) + f_{ij}(T^{j-1,j}(x))- 4f_{ij}(x)
\end{equation}
At the left reservoir $i=1, 2<j$ we have
\[
\mathcal{L}_L(f_{1,j}(x)) = f_{1j}(T^{1,2}(x)) + f_{1j}(T^{0,1}(x)) + f_{1j}(T^{j,j+1}(x)) + f_{1j}(T^{j-1,j}(x)) - 4f_{1j}(x)
\]
and for the right boundary $j=N, i< N-1$
\[
\mathcal{L}_R(f_{i,N}(x)) = f_{iN}(T^{i,i+1}(x)) + f_{iN}(T^{i-1,i}(x)) + f_{iN}(T^{N,N+1}(x)) + f_{iN}(T^{N-1,N}(x)) - 4f_{iN}(x)
\]

It follows for  $i < j$, $i=2,...,N-2$ and $j=4,...,N-1 $ that the generator $\mathcal{L}$ acting on $f_{ij}(x)$ with $|i-j|>1$ can be written as
\[
\mathcal{L}(f_{ij}(x)) = \mathcal{L}_L(f_{1j}(x)) + \mathcal{L}_b(f_{ij}(x)) + \mathcal{L}_R(f_{iN}(x))
\]

\begin{equation}\label{GeniFarj}
\begin{split}
 \mathcal{L}_L(f_{1j}(x)) &= \frac{(1-\lambda)}{2}\biggl[f_{2,j}(x)+ T_L f_{j}(x) +f_{1,j+1}(x) + f_{1,j-1}(x)-4f_{1,j}(x)\biggr]  \\
\mathcal{L}_b(f_{ij}(x)) &= \frac{(1-\lambda)}{2} \biggl[f_{i+1,j}(x) + f_{i-1,j}(x) + f_{i,j+1}(x) + f_{i,j-1}(x)- 4 f_{i,j}(x)\biggr] \\
\mathcal{L}_R(f_{iN}(x)) &= \frac{(1-\lambda)}{2}\biggl[f_{i+1,N}(x)+f_{i-1,N}(x) +f_{i}(x) T_R + f_{i,N-1}(x)-4f_{i,N}(x)\biggr]
\end{split}
\end{equation}
%%%%%%%%%%%%%%%%%%%%%%%%%%%%%%%

\subsubsection{Case III: $|i-j|=1$}
Finally we calculate $\mathcal{L}(f_{i,i+1}(x))$ for the off-diagonal elements $f_{i,i+1}(x)=x_ix_{i+1}$. In this case
\[
\mathcal{L}(f_{i,i+1}(x)) = \mathcal{L}_L(f_{1,2}(x)) + \mathcal{L}_b(f_{i,i+1}(x)) + \mathcal{L}_R(f_{N-1,N}(x))
\]
Fix $i=2,...,N-2$. It is easy to verify that
\begin{equation}
 \mathcal{L}_b(f_{i,i+1}(x))  =  f_{i,i+1}(T^{i,i+1}(x)) + f_{i,i+1}(T^{i-1,i}(x)) + f_{i,i+1}(T^{i+1,i+2}(x))- 3f_{i,i+1}(x)
\end{equation}
At the left boundary for $i=1$ the generator acting on $f_{12}(x)$ is given by
\[
\mathcal{L}_L(f_{12}(x)) = f_{12}(T^{1,2}(x)) + f_{12}(T^{0,1}(x))+ f_{12}(T^{2,3}(x)) - 3f_{12}(x)
\]
and for $i=N-1$
\begin{equation*}
\begin{split}
\mathcal{L}_R(&f_{N-1,N}(x)) = \\
& f_{N-1,N}(T^{N-1,N}(x)) + f_{N-1,N}(T^{N-2,N-1}(x))+ f_{N-1,N}(T^{N,N+1}(x)) - 3f_{N-1,N}(x)
\end{split}
\end{equation*}
which yields
\begin{equation}\label{GeniMinusj}
\begin{split}
 \mathcal{L}_L(&f_{1,2}(x))  =
(1-\lambda)\biggl[ C(f_{11}(x) + f_{22}(x)) +\frac{1}{2}(f_{13}(x)+T_L f_2(x))-(1+B)f_{12}(x)\biggr] \\
 \mathcal{L}_b(&f_{i,i+1}(x))  = \\
& (1-\lambda)\biggl[C(f_{ii}(x) + f_{i+1,i+1}(x)) + \frac{1}{2}(f_{i-1,i+1}(x) + f_{i,i+2}(x) )- (1+B)f_{i,i+1}(x) \biggr]\\
 \mathcal{L}_R(&f_{N-1,N}(x))  = \\
&(1-\lambda) \biggl[ C(f_{N-1,N-1}(x) + f_{NN}(x)) +\frac{1}{2}(f_{N-2,N}(x)+T_R f_{N-1}(x))-(1+B)f_{N-1,N}(x)\biggr]\\
\end{split}
\end{equation}
and
\[
C := C(\alpha,\lambda)=\frac{\lambda}{2} + \alpha(1-\lambda)
\]

 \end{document}